\documentclass[11pt]{article}

\usepackage[USenglish]{babel} 
\usepackage[T1]{fontenc}      
\usepackage[utf8]{inputenc}   
\usepackage{graphicx}
\usepackage{mathrsfs}
\usepackage{amsmath}
\usepackage{amsthm}
\usepackage{amssymb}
\usepackage{authblk}
\usepackage{mathtools}
\usepackage{braket}
\usepackage{framed}

\usepackage{cleveref}
\usepackage{float}

\usepackage{graphicx}
\usepackage{subcaption} 

\usepackage{url}
\usepackage{xcolor}

\usepackage[textsize=scriptsize]{todonotes}

\usepackage{xspace}

\makeatletter
\newcommand*{\etc}{%
    \@ifnextchar{.}%
        {etc}%
        {etc.\@\xspace}%
}
\makeatother

\makeatletter
\newcommand*{\etal}{%
    \@ifnextchar{.}%
        {et al}%
        {et al.\@\xspace}%
}
\makeatother

\numberwithin{equation}{section}

\newtheorem{theorem}{Theorem}[section]
\crefname{theorem}{Theorem}{Theorems}

\newtheorem{lemma}[theorem]{Lemma}
\crefname{lemma}{Lemma}{Lemmas}

\crefname{equation}{equation}{equations}

\newtheorem{observation}[theorem]{Observation}
\newtheorem*{observation*}{Observation}

\usepackage{upgreek}
\usepackage{mathtools}
\usepackage{stmaryrd}
\usepackage{flexisym} 
\usepackage[textsize=scriptsize]{todonotes}

\DeclarePairedDelimiter\abs{\lvert}{\rvert}

\DeclarePairedDelimiter\norm{\lVert}{\rVert}

\newcommand{\N}{\mathbb{N}}
\newcommand{\Q}{\mathbb{Q}}
\newcommand{\R}{\mathbb{R}}
\newcommand{\Z}{\mathbb{Z}}

\newcommand{\sstar}{\Sigma^*}
\newcommand{\sinf}{\Sigma^{\omega}}
\newcommand{\sfi}{\Sigma^{\leq\omega}}
\newcommand{\tri}{\Delta}
\newcommand{\trio}{\tri(\Omega)}
\newcommand{\tripo}{\tri^+(\Omega)}
\newcommand{\dab}{D(\alpha||\beta)}

\newcommand{\shpbox}{\#_{\boxempty}} 
\newcommand{\pisn}{\uppi_{S,n}}        
\newcommand{\pisnl}{\uppi_{S,n}^{(\ell)}} 
\newcommand{\pisnkl}{\uppi_{S,n_k}^{(\ell)}} 
\newcommand{\pitnl}{\uppi_{T,n}^{(\ell)}} 
\newcommand*{\restrict}{\upharpoonright}

\renewcommand{\emph}[1]{\textit{#1}}
\DeclareMathOperator{\divg}{div}
\DeclareMathOperator{\Divg}{Div}
\newcommand*{\divsal}{\divg_{\ell}(S||\alpha)}
\newcommand*{\Divsal}{\Divg_{\ell}(S||\alpha)}

\newcommand*{\divst}{\divg_{\ell}(S||T)}
\newcommand*{\Divst}{\Divg_{\ell}(S||T)}
\newcommand*{\divpl}[3]{\divg_{#3}(#1||#2)}
\newcommand*{\Divpl}[3]{\Divg_{#3}(#1||#2)}
\newcommand*{\divp}[2]{\divg(#1||#2)}
\newcommand*{\Divp}[2]{\Divg(#1||#2)}

\newcommand*{\anorm}{{$\alpha$-normal} }

\newcommand*{\prefix}{\sqsubseteq}
\newcommand*{\sprefix}{\sqsubset} 
\newcommand*{\dga}{d_{G,\alpha}}
\newcommand*{\al}{\alpha^{(\ell)}}
\DeclareMathOperator{\risk}{risk}
\DeclareMathOperator{\Risk}{Risk}
\DeclareMathOperator{\str}{str}
\DeclareMathOperator{\fs}{FS}
\DeclareMathOperator{\Dim}{Dim}
\newcommand*{\shpgw}{\#_{G,w}}

\newtheorem*{definition*}{Definition}

\title{
    Asymptotic Divergences and Strong Dichotomy
    \thanks{This research was supported in part by National Science Foundation Grants 1247051, 1545028, and 1900716. }
    \thanks{Part of this work was  supported by a grant from the Spanish Ministry of Science, Innovation and Universities (TIN2016-80347-R) and was partly done during a research stay at the Iowa State University supported by National Science Foundation Research Grant 1545028.}
}

\author[1,3]{Xiang Huang}
\author[1]{Jack H. Lutz}
\author[2]{Elvira Mayordomo}
\author[1]{Donald M. Stull}
\affil[1]{Iowa State University, Ames, IA 50011 USA\\
    \texttt{\{huangx, lutz, dstull\}@iastate.edu}}
\affil[2]{Departamento de Inform\'atica e Ingenier\'{\i}a de Sistemas, Instituto de Investigaci\'on en Ingenier\'{\i}a
de Arag\'on, Universidad de Zaragoza, 50018 Zaragoza, Spain\\
    \texttt{elvira@unizar.es}}
\affil[3]{Le Moyne College, Syracuse, NY 13214, USA}
 
\date{}

\begin{document}

\maketitle
\begin{abstract}
The \emph{Schnorr-Stimm dichotomy theorem} \cite{SchSti72} concerns finite-state gamblers that bet on infinite sequences of symbols taken from a finite alphabet $\Sigma$. The theorem asserts that, for any such sequence $S$, the following two things are true. 

    (1) If $S$ is not normal in the sense of Borel (meaning that every two strings of equal length appear with equal asymptotic frequency in $S$), then there is a finite-state gambler that wins money at an infinitely-often exponential rate betting on $S$.

    (2) If $S$ is normal, then any finite-state gambler betting on $S$ loses money at an exponential rate betting on $S$.

In this paper we use the Kullback-Leibler divergence to formulate the \textit{lower asymptotic divergence} $\divg(S||\alpha)$ of a probability measure $\alpha$ on $\Sigma$ from a sequence $S$ over $\Sigma$ and the \textit{upper asymptotic divergence} $\Divg(S||\alpha)$ of $\alpha$ from $S$ in such a way that a sequence $S$ is $\alpha$-normal (meaning that every string $w$ has asymptotic frequency $\alpha(w)$ in $S$) if and only if $\Divg(S||\alpha)=0$.  We also use the Kullback-Leibler divergence to quantify the \textit{total risk} $\Risk_G(w)$ that a finite-state gambler $G$ takes when betting along a prefix $w$ of $S$.

Our main theorem is a \textit{strong dichotomy theorem} that uses the above notions to \textit{quantify} the exponential rates of winning and losing on the two sides of the Schnorr-Stimm dichotomy theorem (with the latter routinely extended from normality to $\alpha$-normality).  Modulo asymptotic caveats in the paper, our strong dichotomy theorem says that the following two things hold for prefixes $w$ of $S$.

    (1\textprime) The infinitely-often exponential rate of winning in 1 is $2^{\Divg(S||\alpha)|w|}$.

    (2\textprime) The exponential rate of loss in 2 is $2^{-\Risk_G(w)}$. 

We also use (1{\textprime}) to show that $1-\Divg(S||\alpha)/c$, where $c= \log(1/ \min_{a\in\Sigma}\alpha(a))$, is an upper bound on the finite-state $\alpha$-dimension of $S$ and prove the dual fact that $1-\divg(S||\alpha)/c$ is an upper bound on the finite-state strong $\alpha$-dimension of $S$.
\end{abstract} 

\section{Introduction}\label{sec:introduction}
An infinite sequence $S$ over a finite alphabet is \textit{normal} in the 1909 sense of Borel \cite{emile1909probabilites} if every two strings of equal length appear with equal asymptotic frequency in $S$. Borel normality played a central role in the origins of measure-theoretic probability theory \cite{billingsley1995probability} and is intuitively regarded as a weak notion of randomness. For a masterful discussion of this intuition, see section 3.5 of \cite{knuth2014art}, where Knuth calls normal sequences ``$\infty$-distributed sequences.''

The theory of computing was used to make this intuition precise. This took place in three steps in the 1960s and 1970s. First, Martin-L\"of \cite{martin1966definition} used constructive measure theory to give the first successful formulation of the randomness of individual infinite binary sequences. Second, Schnorr \cite{DBLP:journals/mst/Schnorr71}  gave an equivalent, and more flexible, formulation of Martin-L\"of's notion in terms of gambling strategies called \textit{martingales}. In this formulation, an infinite binary sequences $S$ is random if no lower semicomputable martingale can make unbounded money betting on the successive bits of $S$. Third, Schnorr and Stimm \cite{SchSti72} proved that an infinite binary sequence $S$ is normal if and only if no martingale that is computed by a finite-state automaton can make unbounded money betting on the successive bits of $S$. That is, \textit{normality is finite-state randomness.}

This equivalence was a breakthrough that has already had many consequences (discussed later in this introduction), but the Schnorr-Stimm result said more. It is a \textit{dichotomy theorem}  asserting that, for any infinite binary sequence $S$, the following two things are true.
\begin{enumerate}
	\item If $S$ is not normal, then there is a finite-state gambler that makes money at an infinitely-often exponential rate when betting on $S$.
	\item If $S$ is normal, then every finite-state gambler that bets infinitely many times on $S$ loses money at an exponential rate. 
\end{enumerate}

The main contribution of this paper is to \emph{quantify} the exponential rates of winning and losing on the two sides (1 and 2 above) of the Schnorr-Stimm dichotomy. 

To describe our main theorem in some detail, let $\Sigma$ be a finite alphabet. It is routine to extend the above notion of normality to an arbitrary probability measure $\alpha$ on $\Sigma$. Specifically, an infinite sequence $S$ over $\Sigma$ is $\alpha$-\textit{normal} if every finite string $w$ over $\Sigma$ appears with asymptotic frequency $\alpha^{|w|}(w)$ in $S$, where $\alpha^{\ell}$ is the natural (product) extension of $\alpha$ to strings of length $\ell$. Schnorr and Stimm \cite{SchSti72} correctly noted that their dichotomy theorem extends to $\alpha$-normal sequences in a straightforward manner, and it is this extension whose exponential rates we quantify here. 

The quantitative tool that drives our approach is the Kullback-Leibler divergence \cite{kullback1951information}, also known as the relative entropy \cite{oCovTho06}. If $\alpha$ and $\beta$ are probability measures on $\Sigma$, then the \textit{Kullback-Leibler divergence of $\beta$ from $\alpha$} is 
\[
 D(\alpha||\beta)=E_{\alpha}\log\frac{\alpha}{\beta},   	
\]   
i.e., the expectation with respect to $\alpha$ of the random variable $\log\frac{\alpha}{\beta}:\Sigma \rightarrow \R\cup\{\infty\}$, where the logarithm is base-2. Although the Kullback-Leibler divergence is not a metric on the space of probability measures on $\Sigma$, it does quantify ``how different'' $\beta$ is from $\alpha$, and it has the crucial property that $D(\alpha||\beta)\geq 0$, with equality if and only if $\alpha=\beta$.

Here we use the empirical frequencies of symbols in $S$ to define the \textit{asymptotic lower divergence} $\divg(S||\alpha)$ of $\alpha$ from $S$ and the \textit{asymptotic upper divergence} $\Divg(S||\alpha)$ of $\alpha$ from $S$ in a natural way, so that $S$ is $\alpha$-normal if and only if $\Divg(S||\alpha)=0$.

The first part of our \textit{strong dichotomy theorem} says that the infinitely-often exponential rate that can be achieved in $1$ above is essentially at least $2^{\Divg(S||\alpha)\abs{w}}$, where $w$ is the prefix of $S$ on which the finite-state gambler has bet so far. More precisely, it says the following. 

\begin{enumerate}
 	\item[1\textprime.] If $S$ is not $\alpha$-normal, then, for every $\gamma<1$, there is a finite-state gambler $G$ such that, when $G$ bets on $S$ with payoffs according to $\alpha$, there are infinitely many prefixes $w$ of $S$ after which $G$'s capital exceeds $2^{\gamma \Divg(S||\alpha)\abs{w}}$. 
 \end{enumerate} 

 The second part of our strong dichotomy theorem, like the second part of the Schnorr-Stimm dichotomy theorem, is complicated by the fact that a finite-state gambler may, in some states, decline to bet. In this case, its capital after a bet is the same as it was before the bet, regardless of what symbol actually appears in $S$. Once again, however, it is the Kullback-Leibler divergence that clarifies the situation. As explained in section 3 below, in any particular state $q$, a finite-state gambler's betting strategy is a probability measure $B(q)$ on $\Sigma$. If $B(q)=\alpha$, then the gambler does not bet in state $q$. We thus define the \textit{risk} that the gambler $G$ takes in state $q$ to be 
 \[
  	\risk_G(q)=D(\alpha|| B(q)),
  \] 
  i.e., the divergence of $B(q)$ from not betting. We then define the \textit{total risk} that the gambler takes along a prefix $w$ of the sequence $S$ on which it is betting to be the sum $\Risk_{G}(w)$ of the risks $\risk_G(q)$ in the states that $G$ traverses along $w$. The second part of our strong dichotomy theorem says that, if $S$ is $\alpha$-normal and $G$ is a finite-state gambler betting on $S$, then after each prefix $w$ of $S$, the capital of $G$ on prefixes $w$ of $S$ is essentially bounded above by $2^{-\Risk_{G}(w)}$. In some sense, then, $G$ loses all that it risks. More precisely, the second part of our strong dichotomy says the following.

  \begin{enumerate}
  	\item[2\textprime.]  If $S$ is $\alpha$-normal, then, for every finite-state gambler $G$ and every $\gamma<1$, after all but finitely many prefixes $w$ of $S$, the gambler $G$'s capital is less than $2^{-\gamma \Risk_{G}(w)}$. 
  \end{enumerate}

  A routine ergodic argument, already present in \cite{SchSti72}, shows that, if a finite-state gambler $G$ bets on an \anorm sequence $S$, then every state of $G$ that occurs infinitely often along $S$ occurs with positive frequency along $S$. Hence 2 above follows from 2{\textprime} above.

  Our strong dichotomy theorem has implications for finite-state dimensions. For each probability measure $\alpha$ on $\Sigma$ and each sequence $S$ over $\Sigma$, the \textit{finite-state $\alpha$-dimension} $\dim_{\fs}^{\alpha}(S)$ and the \textit{finite-state strong $\alpha$-dimension} $\Dim_{\fs}^{\alpha}(S)$ (defined in section 4 below) are finite-state versions of Billingsley dimension \cite{billingsley1960hausdorff,cajar1981billingsley} introduced in \cite{lutz11}. When $\alpha$ is the uniform probability measure on $\Sigma$, these are the finite dimension $\dim_{\fs}(S)$, introduced in \cite{DLLM04} as a finite-state version of Hausdorff dimension \cite{haus19,falc14}, and the finite-state strong dimension $\Dim_{\fs}(S)$, introduced in \cite{AHLM07} as a finite-state version of packing dimension \cite{tric1982,sullivan1984entropy,falc14}. Intuitively, $dim_{\fs}^{\alpha}(S)$ and $\Dim_{\fs}^{\alpha}(S)$ measure the lower and upper asymptotic $\alpha$-densities of the finite-state information in $S$. 

  Here we use part 1 of our strong dichotomy theorem to prove that, for every positive probability measure $\alpha$ on $\Sigma$ and every sequence $S$ over $\Sigma$,
  \[
  	\dim_{\fs}^{\alpha}(S)\leq 1-\Divg(S||\alpha)/c,
  \]
  where $c= \log(1/ \min_{a\in\Sigma}\alpha(a))$. We also establish the dual result that, for all such $\alpha$ and $S$,
  \[
  	\Dim_{\fs}^{\alpha}(S)\leq 1-\divg(S||\alpha)/c.
  \]

Research on normal sequences and normal numbers (real numbers whose base-$b$ expansions are normal sequences for various choices of $b$) 
has grown rapidly in recent years.  Part of this is due to the fact that Agafonov \cite{agafonov1968normal}\ and Schnorr and Stimm \cite{SchSti72}\
connected the theory of normal numbers so directly to the theory of computing.  Further work along these lines has been continued in 
\cite{KamWei75,MerRei06,BecHei13,shen2017automatic}.  
After the discovery of algorithmic dimensions in the present century \cite{lutz03a,jLutz03, DLLM04,AHLM07}, the Schnorr-Stimm dichotomy led to the realization 
\cite{BoHiVi05}\ that the finite-state world, unlike any other known to date, is one in which maximum dimension is not only necessary, but also sufficient, 
for randomness.  This in turn led to the discovery of nontrivial extensions of classical theorems on normal numbers \cite{CopErd46,Wall49}\ to new quantitative 
theorems on finite-state dimensions \cite{GuLuMo07,DoLuNa07}, a line of inquiry that will certainly continue.  It has also led to a polynomial-time algorithm 
\cite{BeHeSl13}\ that computes real numbers that are provably absolutely normal (normal in every base) and, via Lempel-Ziv methods, 
to a nearly linear time algorithm for this \cite{LutMay18}.  In parallel with these developments, connections among normality, Weyl equidistribution theorems, 
and Diophantine approximation have led to a great deal of progress surveyed in the books \cite{DajKra02,Buge12}.  This paragraph does not begin to do justice to the 
breadth and depth of recent and ongoing research on normal numbers and their growing involvement with the theory of computing.  It is to be hoped that our strong 
dichotomy theorem and the quantitative methods implicit in it will further accelerate these discoveries.

\section{Divergence and normality}\label{sec:divergence_and_normality}
This section reviews the discrete Kullback-Leibler divergence, introduces asymptotic extensions of this divergence, and uses these to give useful characterizations of Borel normal sequences.
\subsection{The Kullback-Leibler divergence}\label{subs:KL_divergence}
We work in a finite alphabet $\Sigma$ with $2\leq \abs{\Sigma}<\infty$. We write $\Sigma^\ell$ for the set of strings of length $\ell$ over $\Sigma$, $\sstar=\bigcup_{\ell=0}^\infty\Sigma^{\ell}$ for the set of (finite) \emph{strings} over $\Sigma$, $\Sigma^{\omega}$ for the set of (infinite) \emph{sequences} over $\Sigma$, and $\Sigma^{\leq\omega}=\Sigma^{*}\cup\Sigma^{\omega}$. We write $\lambda$ for the empty string, $\abs{w}$ for the length of a string $w\in\sstar$, and $\abs{S}=\omega$ for the length of a sequence $S\in \sinf$. For $x\in \sfi$ and $0\leq i< \abs{x}$, we write $x[i]$ for the $i$-th symbol in $x$, noting that $x[0]$  is the leftmost symbol in $x$. For $x\in \sfi$ and $0\leq i\leq j< \abs{x}$, we write $x[i..j]$ for the string consisting of the $i$-th through $j$-th symbols in $x$.

A \emph{(discrete) probability measure} on a nonempty finite set $\Omega$ is a function $\pi:\Omega\to[0,1]$ satisfying 
\begin{equation}\label{eq:def_prob_meas}
     \sum_{\omega\in \Omega}\pi(\omega)=1.
 \end{equation} 
 We write $\tri(\Omega)$ for the set of all probability measures on $\Omega$, $\tri^+(\Omega)$ for the set of all $\pi\in\tri(\Omega)$ that are \emph{strictly positive} (i.e., $\pi(\omega)>0$ for all $\omega\in \Omega$), $\tri_
 \Q(\Omega)$ for the set of all $\pi\in\tri(\Omega)$ that are rational-valued, and $\tri_{\Q}^{+}(\Omega)=\tri^{+}(\Omega)\cap\tri_{\Q}(\Omega)$. In this paper we are most interested in the case where $\Omega=\Sigma^{\ell}$ for some $\ell\in \Z^{+}$.  

 \begin{figure}
\centering
\begin{subfigure}{.5\textwidth}
  \centering
  \includegraphics[width=1\linewidth]{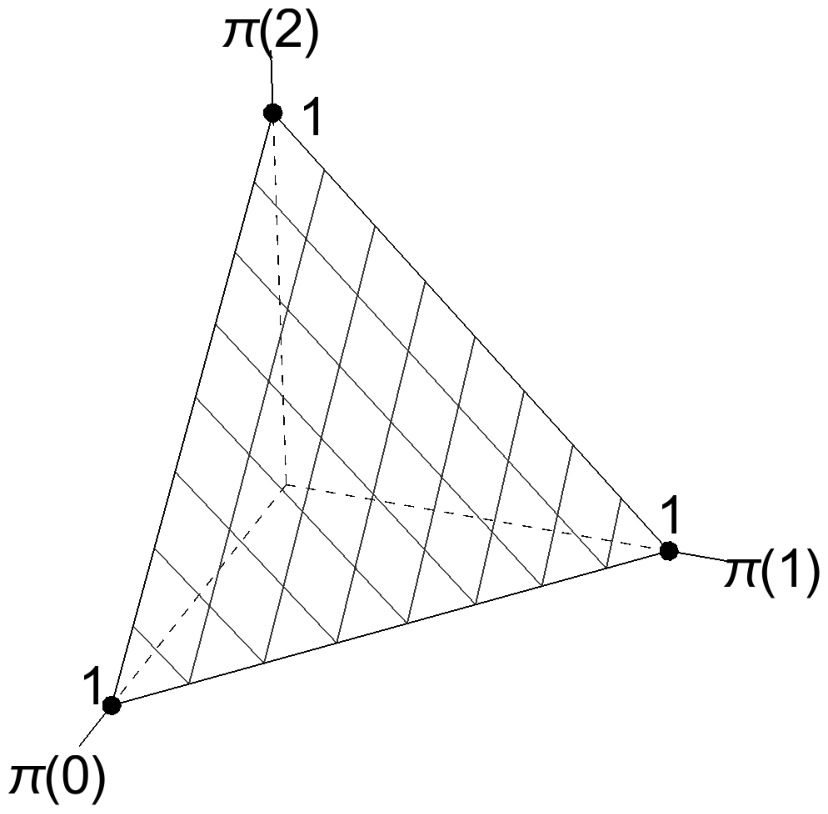}
  \label{fig:sub1}
\end{subfigure}%
\begin{subfigure}{.5\textwidth}
  \centering
  \includegraphics[width=0.8\linewidth]{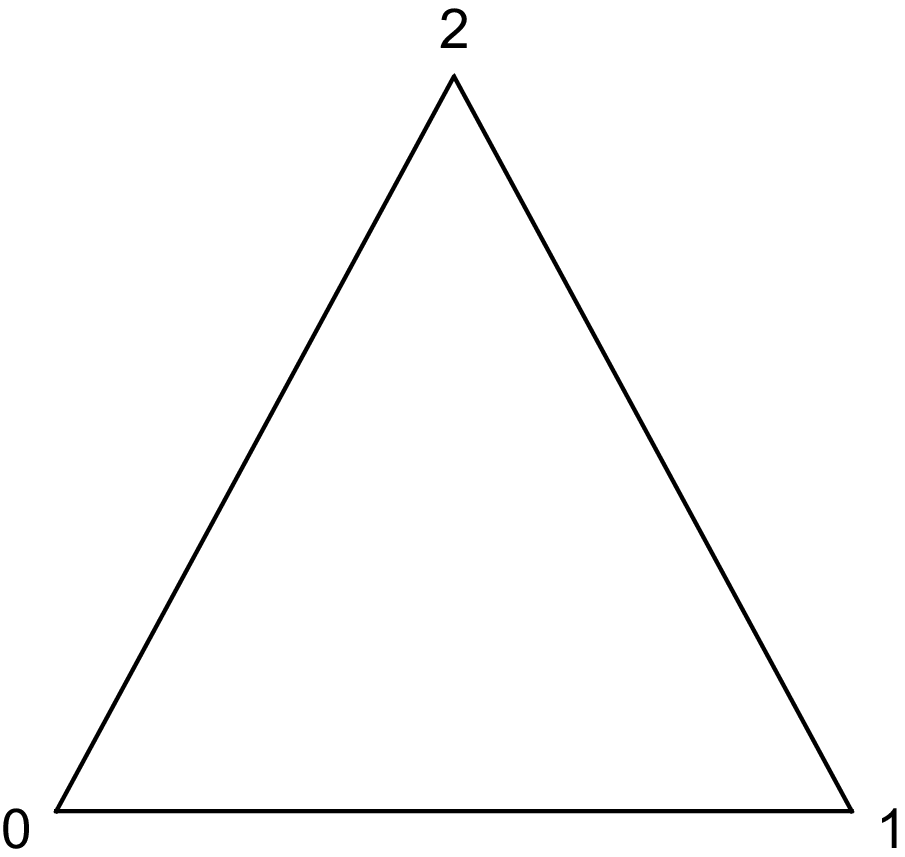}
  \label{fig:sub2}
\end{subfigure}
\caption{Two views of the simplex $\tri(\{0,1,2\})$}
\label{fig:prob_measures}
\end{figure}

 Intuitively, we identify each probability measure $\pi\in\tri(\Omega)$ with the point in $\R^{\abs{\Omega}}$ whose coordinates are the probabilities $\pi(\omega)$ for $\omega\in\Omega$. By (\ref{eq:def_prob_meas}) this implies that $\tri(\Omega)$ is the $(\abs{\Omega}-1)$-dimensional simplex in $\R^{\abs{\Omega}}$ whose vertices are the points at 1 on each of the coordinate axes. (See Figure 1 for an illustration with $|\Omega|=3$.) For each $\omega \in \Omega$, the vertex on axis $\omega$ is the degenerate probability measures $\pi_{\omega}$ with $\pi_{\omega}(\omega)=1$. The centroid of the simplex $\tri(\Omega)$ is the uniform probability measure on $\Omega$, and the (topological) interior of $\trio$ is $\tripo$. We write $\partial\trio=\trio\smallsetminus\tripo$ for the boundary of $\trio$.

    \begin{definition*}(\cite{kullback1951information}).
    Let $\alpha,\beta\in\trio$, where $\Omega$ is a nonempty finite set. The \emph{Kullback-Leibler divergence} (or \emph{KL-divergence}) \emph{of $\beta$ from $\alpha$} is 
\begin{equation}\label{eq:def_KL_divergence}
    D(\alpha||\beta)=E_\alpha \log{\frac{\alpha}{\beta}},
\end{equation}
where the logarithm is base-2.
\end{definition*}

Note that the right-hand side of (\ref{eq:def_KL_divergence}) is the $\alpha$-expectation of the random variable 
\begin{equation*}
    \log{\frac{\alpha}{\beta}}:\Omega\longrightarrow\R
\end{equation*}
defined by 
\[
    \left(\log\frac{\alpha}{\beta}\right)(\omega)=\log\frac{\alpha(\omega)}{\beta(\omega)}
\]
for each $\omega\in \Omega$. Hence (\ref{eq:def_KL_divergence}) is a convenient shorthand for 
\[
    D(\alpha||\beta)=\sum_{\omega\in\Omega}\alpha(\omega)\log\frac{\alpha(\omega)}{\beta(\omega)}.
\]
Note also that $D(\alpha||\beta)$ is infinite if and only if $\alpha(\omega)>0=\beta(\omega)$ the some $\omega\in\Omega$.

The Kullback-Leibler divergence $D(\alpha||\beta)$ is a useful measure of how different $\beta$ is from $\alpha$. It is not a metric (because it is not symmetric and does not satisfy the triangle inequality), but it has the crucial property that $D(\alpha||\beta)\geq0$, with equality if and only if $\alpha=\beta$. The two most central quantities in Shannon information theory, entropy and mutual information, can both be defined in terms of divergence as follows. 
\begin{enumerate}
    \item \emph{Entropy is divergence from certainty.} The \emph{entropy} of a probability measure $\alpha\in \trio$, conceived by Shannon \cite{shannon1948mathematical} as a measure of the uncertainty of $\alpha$, is    
    \begin{equation}\label{eq:entropy_by_divergence}
        H(\alpha)=\sum_{\omega\in\Omega}\alpha(\omega)D(\pi_{\omega}||\alpha),
    \end{equation}
    i.e., the $\alpha$-average of the divergences of $\alpha$ from the ``certainties'' $\pi_{\omega}$.
    \item \emph{Mutual information is divergence from independence.} If $\alpha,\ \beta\in\trio$ have a joint probability measure $\gamma\in \tri(\Omega\times\Omega)$ (i.e., are the marginal probability measures of $\gamma$), then the \emph{mutual information} between $\alpha$ and $\beta$, conceived by Shannon \cite{shannon1948mathematical} as a measure of the information shared by $\alpha$ and $\beta$, is 
    \begin{equation}
        I(\alpha\,;\,\beta)=D(\alpha\beta||\gamma),
    \end{equation}
    i.e., the divergence of $\gamma$ from the probability measure in which $\alpha$ and $\beta$ are independent.     
\end{enumerate}

    Two additional properties of the Kullback-Leibler divergence are useful for our asymptotic concerns. First, the divergence $\dab$ is continuous on $\trio^2$ (as a function  into $[0,\infty]$). Hence, if $\alpha_n\in \trio$ for each $n\in \N$ and $\displaystyle\lim_{n\to\infty}\alpha_n=\alpha$ in the sense of the Euclidean metric on the simplex $\trio$, then $\displaystyle\lim_{n\to \infty} D(\alpha_n||\alpha)=\lim_{n\to\infty}D(\alpha||\alpha_n)=0$. Second, the converse holds. It is well known \cite{oCovTho06} that \[\dab\geq\frac{1}{2\ln 2}\norm{\alpha-\beta}_1^2,\] where 
    \(
        \norm{\alpha-\beta}_1=\sum_{\omega\in\Omega}\abs{\alpha(\omega)-\beta(\omega)}
    \)
    is the $\mathscr{L}_1$-norm. Hence, if either $\displaystyle\lim_{n\to\infty}D(\alpha_n||\alpha)=0$ or $\displaystyle\lim_{n\to\infty}D(\alpha||\alpha_n)=0$, then $\displaystyle\lim_{n\to\infty}\alpha_n=\alpha$.

    More extensive discussions of the Kullback-Leibler divergence appear in \cite{oCovTho06,csiszar2004information}.

\subsection{Asymptotic divergences} 
\label{sub:asymptotic_divergences}
For nonempty strings $w,x\in \sstar$, we write 
\[
    \shpbox(w,x)=\left\vert\big\{ m\leq \frac{\abs{x}}{\abs{w}}-1 \ \bigg\vert \ x[m\abs{w}..(m+1)\abs{w}-1]=w \big\}\right\vert
\]\
for the \emph{number of block occurrences} of $w$ in $x$. Note that $0\leq \shpbox(w,x)\leq \frac{\abs{x}}{\abs{w}}$.

For each $S\in\sinf$, $n\in \Z^+$, and $\lambda\not=w\in\sstar$, the $n$-th \emph{block frequency} of $w$ in $S$ is 
\begin{equation}\label{def:block_freq}
     \pisn(w)=\frac{\shpbox(w,S[0..n|w|-1])}{n} 
 \end{equation} 
 Note that (\ref{def:block_freq}) defines, for each $S\in \sinf$ and $n\in \Z^+$, a function
 \[
     \pisn:\sstar\smallsetminus\{\lambda\}\longrightarrow\Q.
 \]
 For each such $S$ and $n$ and each $\ell \in \Z^+$, let $\pisnl=\pisn\restrict\Sigma^{\ell}$ be the restriction of the function $\pisn$ to the set $\Sigma^\ell$ of strings of length $\ell$. 

 \begin{observation}
 For each $S\in\sinf$ and $n,\,\ell\in \Z^+$,
 \[
     \pisnl\in\tri_{\Q}(\Sigma^{\ell}),
 \]
 i.e., $\pisnl$ is a rational-valued probability measure on $\Sigma^\ell$.
 \end{observation}

We call $\pisnl$ the $n$-th \emph{empirical probability measure on $\Sigma^{\ell}$ given by $S$}.

A probability measure $\alpha\in\tri(\Sigma)$ naturally induces, for each $\ell\in \Z^+$, a probability measure $\alpha^{(\ell)}\in \tri(\Sigma^{\ell})$ defined by 
\begin{equation}
     \alpha^{(\ell)}(w)=\prod_{i=0}^{\abs{w}-1}\alpha(w[i]).
 \end{equation} 

 The empirical probability measures $\pisnl$ provide a natural way to define useful empirical divergences of probability measures from sequences.

 \begin{definition*}
    Let $\ell\in\Z^+,$ $S\in\sinf$, and $\alpha\in\tri(\Sigma)$.
    \begin{enumerate}
         \item The \emph{lower $\ell$-divergence of $\alpha$ from} $S$ is 
         $\displaystyle \divsal=\liminf_{n\to\infty}D(\pisnl||\alpha^{(\ell)}).$
          \item The \emph{upper $\ell$-divergence of $\alpha$ from $S$} is 
           $\displaystyle \Divsal=\limsup_{n\to\infty}D(\pisnl||\alpha^{(\ell)}).$
          \item The \emph{lower divergence of $\alpha$ from $S$} is 
              $\displaystyle\divp{S}{\alpha}=\sup_{\ell\in\Z^+}\divsal/\ell$.
          \item The \emph{upper divergence of $\alpha$ from $S$} is 
           $\Divp{S}{\alpha}=\sup_{\ell\in\Z^+}\Divsal/\ell$.
     \end{enumerate} 
 \end{definition*}
 
 A similar approach gives useful empirical divergences of one sequence from another. 
\begin{definition*}
Let $\ell \in \Z^+$ and $S,\, T\in \sinf$.
\begin{enumerate}
    \item The \emph{lower $\ell$-divergence of $T$ from $S$} is 
        $\displaystyle\divpl{S}{T}{\ell}=\liminf_{n\to\infty}D(\pisnl||\,\pitnl)$.
    \item The \emph{upper $\ell$-divergence of $T$ from $S$} is 
    $\displaystyle\Divpl{S}{T}{\ell}=\limsup_{n\to\infty}D(\pisnl||\,\pitnl).$
    \item The \emph{lower divergence of $T$ from $S$} is 
    $\divp{S}{T}=\sup_{\ell\in \Z^+} \divst/\ell$.
    \item The \emph{upper divergence of $T$ from $S$} is
    $\Divp{S}{T}=\sup_{\ell\in \Z^+} \Divst/\ell$.
\end{enumerate}
\end{definition*}
\subsection{Normality} 
\label{ssub:normality}
The following notions are essentially due to Borel \cite{emile1909probabilites}. 
\begin{definition*}
    Let $\alpha\in\tri(\Sigma)$, $S\in \sinf$, and $\ell\in \Z^+$.
    \begin{enumerate}
        \item $S$ is $\alpha$-$\ell$-\emph{normal} if, for all $w\in \Sigma^{\ell}$,
        \[
             \lim_{n\to\infty}\pisn(w)=\alpha^{(\ell)}(w).
        \]
        \item $S$ is $\alpha$-\emph{normal} if, for all $\ell \in \Z^+$, $S$ is $\alpha$-$\ell$-normal.
        \item $S$ is $\ell$-\emph{normal} if $S$ is $\mu$-$\ell$-normal, where $\mu$ is the uniform probability measure on $\Sigma$. 
        \item $S$ is \emph{normal} if, for all $\ell\in\Z^+$, $S$ is $\ell$-normal. 
    \end{enumerate}
\end{definition*}
\begin{lemma}\label{lemma:alpha_l_equiv_2.2}
    For all $\alpha\in \tri(\Sigma)$, $S\in \sinf$, and $\ell \in \Z^+$, the following four conditions are equivalent. 
    \begin{enumerate}
        \item[(1)] $S$ is $\alpha$-$\ell$-normal.
        \item[(2)] $\Divsal$=0.
        \item[(3)] For every $\alpha$-$\ell$-normal sequence $T\in \sinf$, $\Divpl{S}{T}{\ell}=0$.
        \item[(4)] There exists an $\alpha$-$\ell$-normal sequence $T\in \sinf$ such that $\Divpl{S}{T}{\ell}=0$.
    \end{enumerate}
\end{lemma}
\begin{proof}
    Let $\alpha,\ S,$ and $\ell$ be as given. 

    To see that (1) implies (2), assume (1). Then $\displaystyle\lim_{n\to\infty}\pisnl=\alpha^{(\ell)}$, so the continuity of KL-divergence tells us that 
    \[
        \Divg_{\ell}(S||\alpha)=\lim_{n\to\infty}D(\pisnl||\alpha^{(\ell)})=0,
    \]
    i.e., that (2) holds. 

    To see that (2) implies (3), assume (2). Then $\displaystyle\lim_{n\to\infty}D(\pisnl||\alpha^{(\ell)})=\ell\Divg_{\ell}(S||\alpha)=0$, whence the $\mathscr{L}_1$ bound in section \ref{subs:KL_divergence} tells us that $\displaystyle\lim_{n\to\infty}\pisnl=\alpha^{(\ell)}$. For any $\alpha$-$\ell$-normal sequence $T\in \sinf$, we have $\displaystyle\lim_{n\to\infty}\pitnl=\alpha^{(\ell)}$, whence the continuity of KL-divergence tells us that 
    \[
         \Divst=\lim_{n\to\infty}D(\pisnl||\pitnl)=D(\alpha^{(\ell)}||\alpha^{(\ell)})=0,
     \] 
     i.e., that (3) holds. 

     Since $\alpha$-$\ell$-normal sequences exist, it is trivial that (3) implies (4). 

     Finally, to see that (4) implies (1), assume that (4) holds. Then we have 
     \[
         \lim_{n\to \infty}D(\pisnl||\pitnl)=\Divst=0,
     \]
     whence the $\mathscr{L}_1$ bound in section \ref{subs:KL_divergence} tells us that
 \begin{equation}\label{eq:ST_l_close}
     \lim_{n\to\infty}\norm{\pisnl-\pitnl}_1=0.
 \end{equation}
     We also have 
     \[
         \lim_{n\to\infty} \pitnl=\alpha^{(\ell)},
     \]
     whence 
     \begin{equation}\label{eq:pi_T_alpha_close}
         \lim_{n\to\infty}\norm{\pitnl-\alpha^{(\ell)}}_1 =0.
     \end{equation}
     By (\ref{eq:ST_l_close}), (\ref{eq:pi_T_alpha_close}), and the triangle inequality for the $\mathscr{L}_1$-norm, we have 
     \[
         \lim_{n\to\infty}\norm{\pisnl-\alpha^{(\ell)}}_1=0,
     \]
     whence 
     \[
        \lim_{n\to\infty}\pisnl=\alpha^{(\ell)}, 
     \]
     i.e., (1) holds.    
\end{proof}
Lemma \ref{lemma:alpha_l_equiv_2.2} immediately implies the following.
                 
\begin{theorem}[divergence characterization of normality]\label{thm:divg_char_normality}                
    For all $\alpha\in \tri(\Sigma)$ and $S\in \sinf$, the following conditions are equivalent.
    \begin{enumerate}
        \item [(1)] $S$ is $\alpha$-normal.
        \item [(2)] $\Divg(S||\alpha)=0$.
        \item [(3)] For every $\alpha$-normal sequence $T\in \sinf$, $\Divp{S}{T}=0$.
        \item [(4)] There exists an $\alpha$-normal sequence $T\in \sinf$ such that $\Divp{S}{T}=0$.
    \end{enumerate}
\end{theorem}                   
\section{Strong Dichotomy}
This section presents our main theorem, the strong dichotomy theorem for finite-state gambling. We first review finite-state gamblers. 

Fix a finite alphabet $\Sigma$ with $\abs{\Sigma}\geq 2$. 

\begin{definition*} [\cite{SchSti72,Fede91,DLLM04}]
	A \emph{finite-state gambler (FSG)} is a 4-tuple
	\[
		G=(Q,\delta,s,B),
	\]
where $Q$ is a finite set of \emph{states}, $\delta:Q\times \Sigma\rightarrow Q$ is the \emph{transition function}, $s\in Q$ is the \emph{start state}, and $B:Q\rightarrow \tri_{\Q}(\Sigma)$ is the \emph{betting function}.
\end{definition*}

The transition structure $(Q,\delta, s)$ here works as in any deterministic finite-state automaton. For $w\in \sstar$, we write $\delta(w)$ for the state reached from $s$ by processing $w$.

Intuitively, a gambler $G=(Q,\delta,s, B)$ bets on the successive symbols of a sequence $S\in \sinf$. The payoffs in the betting are determined by a \emph{payoff probability measure} $\alpha\in \tri(\Sigma)$.  (We regard $\alpha$ and $S$ as external to the gambler $G$.) We write $d_{G,\alpha}(w)$ for the gambler $G$'s capital (amount of money) after betting on the successive bits of a prefix $w\prefix S$, and we assume that the initial capital is $d_{G,\alpha}(\lambda)=1$.

The meaning of the betting function $B$ is as follows. After betting on a prefix $w\prefix S$, the gambler is in state $\delta(w)\in Q$. The betting function $B$ says that, for each $a\in \Sigma$, the gambler bets the fraction $B(\delta(w))(a)$ of its current capital $d_{G,\alpha}(w)$ that $wa\prefix S$, i.e., that the next symbol of $S$ is an $a$. If it then turns out to be the case that $wa\prefix S$, the gambler's capital will be 
\begin{equation}\label{eq:gambler_def}
	d_{G,\alpha}(wa)=\dga(w)\frac{B\big(\delta(w)\big)(a)}{\alpha(a)}.
\end{equation}
(Note: If $\alpha(a)=0$ here, we may define $\dga(wa)$ however we wish.) 

The payoffs in (\ref{eq:gambler_def}) are fair with respect to $\alpha$, which means that the conditional $\alpha$-expectation 
\[
	\sum_{a\in \Sigma}\alpha(a)\dga(wa)
\]
of $\dga(wa)$, given that $w\prefix S$, is exactly $\dga(w)$. This says that the function $\dga$ is an $\alpha$-\emph{martingale}.

If $\delta(w)=q$ is a state in which $B(q)=\alpha$, then (\ref{eq:gambler_def}) says that, for each $a\in \Sigma$, $\dga(wa)=\dga(w)$. That is, the condition $B(q)=\alpha$ means that $G$ \emph{does not bet} in state $q$. Accordingly, we define the \emph{risk} that $G$ takes in a state $q\in Q$ to be 
\[
	\risk_G(q)=D\big(\alpha ||B(q)\big).
\]
i.e., the divergence of $B(q)$ from not betting. We also define the \emph{total risk} that $G$ takes \emph{along} a string $w\in \sstar$ to be 
\[
	\Risk_G(w)=\sum_{u\sprefix w} \risk_{G}\big(\delta(u)\big).
\]
We now state our main theorem. 

\begin{theorem}[strong dichotomy theorem]
Let $\alpha \in \tri(\Sigma)$, $S\in \sinf$, and $\gamma<1$.
\begin{enumerate}
	\item If $S$ is not $\alpha$-normal, then there is a finite-state gambler $G$ such that, for infinitely many prefixes $w\prefix S$,
	\[
		\dga(w)> 2^{\gamma \Divg(S||\alpha)\abs{w}}. 
	\]
	\item If $S$ is $\alpha$-normal, then, for every finite-state gambler $G$, for all but finitely many prefixes $w\prefix S$,
	\[
		\dga(w)<2^{-\gamma \Risk_G(w)}.
	\]
\end{enumerate}
\end{theorem}
\begin{proof}
	To prove the first part, let $S$ be a non-normal sequence. Then by Theorem \ref{thm:divg_char_normality} we know that $\Divg(S||\alpha)>0$. Let $r<1$ and let $\epsilon>0$.
	By the definition of $\Divg(S||\alpha)$ there must exist $\ell$ such that
	\begin{equation}\label{eq:r_r_prime_Div}
		\Divg_\ell(S||\alpha)/\ell>r\Divg(S||\alpha).
	\end{equation}
	That is 
	\[
		\limsup_{n\to \infty}D(\pisnl||\alpha^{(\ell)})>\ell r\Divg(S||\alpha).
	\]

	We can pick a subsequence of indices $n_k$'s, such that $\lim_{k\to \infty}D(\uppi_{S,n_k}^{(\ell)}||\alpha^{(\ell)})=\Divg_\ell(S||\alpha)$. Therefore by inequality (3.2)
      \begin{equation}\label{ineq:pick_nk}
		D(\pisnkl || \alpha^{(\ell)})>\ell r\Divg(S||\alpha)
	\end{equation}
	for sufficiently large $k$. In particular, by compactness of $\R^{\abs{\Sigma^{\ell}}}$ equipped with $\mathscr{L}_1$-norm, we can further request that 
	\begin{equation}\label{eq:limit_exist}
		\lim_{k\to\infty}\pisnkl  \quad \text{exists}.
	\end{equation}

	Let $\uppi_0=\uppi_0(r, m)\in \tri_{\mathbb{Q}}(\Sigma^{\ell})$ be the $m$-th $\uppi_{S,n_k}^{(\ell)}$ that satisfies (\ref{ineq:pick_nk}), indexed by $k$. By the way we define $\uppi_0$, we have
	\begin{equation}\label{eq:way_to_pick_pi_zero}
		D(\uppi_{S, n_k}||\alpha^{(\ell)})>D(\uppi_0||\alpha^{(\ell)})> \ell r\Divg(S||\alpha),
	\end{equation}
	and 
		\begin{equation}
		\norm{\uppi_0 -\pisnkl}\to 0, \quad\text{as $m\to \infty$ and $k\to \infty$},
	\end{equation}
	whence $D$'s continuity in section \ref{subs:KL_divergence} tells us that 
	\begin{equation}\label{eq:divg_tend_to_zero}
		D(\pisnkl||\uppi_0)\to0, \quad\text{as $m\to \infty$ and $k\to \infty$}.
	\end{equation}
	Also note that,
 	\begin{equation}
		\lim_{m\to \infty} D(\uppi_0||\al)=\lim_{k\to \infty}D(\uppi_{S,n_k}^{(\ell)}||\al)= \Divg_{\ell}(S||\alpha)>0.
	\end{equation}
	
	For a fixed $\uppi_0=\uppi_0(r,m)$, by the definition, for any $n_k$ sufficiently large, we have 
	\begin{equation}\label{ineq:divg_away_from_zero}
		D(\uppi_{S, n_k}^{(\ell)}|| \alpha^{(\ell)})>D(\uppi_0||\alpha^{(\ell)})(1-\epsilon)>0.
	\end{equation}

	By doing the above we pick a probability measure $\uppi_0$ that is ``far'' away from $\alpha^{(\ell)}$, we now hard code $\uppi_0$ in a gambler $G=(Q,\delta, s, B)$, where 
\begin{equation*}
\begin{aligned}[c]
Q=\Sigma^{\leq\ell-1},
\end{aligned}
\qquad
\begin{aligned}[c]
\delta(w,a)=\begin{cases}
	wa &\quad \text{if $\abs{wa}<\ell$} \\
	\lambda&\quad \text{if $\abs{wa}=\ell$}
\end{cases}
\end{aligned},
\qquad
\begin{aligned}[c]
 s=\lambda,
\end{aligned}
\qquad and \qquad
\begin{aligned}[c]
 B(w)(a)=\uppi_0(a | w),
\end{aligned}
\end{equation*}
where $\uppi_0(a | w) $ describes the conditional probability (induced by $\uppi_0$) of occurrence of an $a$ after $w\in Q$, and is defined by  $\uppi_0(a | w) =\uppi_0(wa)/\uppi_0(w)$, where for $u\in Q$, the notation $\uppi_0(u)$ is defined recursively by $\uppi_0(w)=\sum_{a\in\Sigma}\uppi_0(wa)$. 

Note that $G$ can be viewed as a gambler gambling on every $\ell$ symbols, in the way that he always ``waits'' until he sees the first $\ell-1$ symbols of a string $u=wa$ 
of length $\ell$, and then bets a fraction of $\uppi_{0}(wa)$ of his capital on the next symbol being an $a$.

Let $u=a_0\cdots a_{\ell -1}$ be in $\Sigma^{\ell}$. The following observation captures the above intuition:
\[
	\frac{B(\lambda)(a_0)\cdots B(u[0..\ell-2])(a_{\ell-1})}{\alpha(a_0)\cdots\alpha(a_{\ell-1})}=\frac{\uppi_0(u)}{\alpha^{(\ell)}(u)}.
\]
Now let $w=S\restrict n_k$ for some $k$. We can view $w$ as
\[
 	w=u_0u_1\cdots u_{n-1}u_n, \quad \text{where $\abs{u_i}=\ell$ for $0\leq i\leq n-1$ and $u_n=a_0\cdots a_m$ with $m < \ell$}.
 \] 
 Then we have 
 \begin{equation}
  	d_{G,\alpha}(w)=\bigg(\prod_{0}^{n-1}\frac{\uppi_{0}(u_i)}{\alpha(u_i)}\bigg)\frac{B(\lambda)(a_0)\cdots B(u_n[0..m-1])(a_{m})}{\alpha(a_0)\cdots\alpha(a_{m})}
 	\geq C_0\prod_{0}^{n-1}\frac{\uppi_{0}(u_i)}{\alpha(u_i)}, \label{eq:constant_factor}	
 \end{equation}
 where $C_0$ is the minimum value of $\frac{B(\lambda)(a_0)\cdots B(u_n[0..m-1])(a_{m})}{\alpha(a_0)\cdots\alpha(a_{\ell-1})}$, where $u_n=a_0\cdots a_m$ ranges over $\Sigma^{<\ell}$. 
 {}
 Taking $\log$ on both sides of (\ref{eq:constant_factor}) we get
 \begin{align}
 	\log d_{G,\alpha}(w)-\log C_0&\geq \sum_{i=0}^{n-1}\log\frac{\uppi_0(u_i)}{\alpha^{(\ell)}(u_i)}=\sum_{\abs{u}=\ell}\shpbox(u,w)\log\frac{\uppi_0(u)}{\alpha^{(\ell)}(u)}, \notag\\
 	&=n\sum_{\abs{u}=\ell}\frac{\shpbox(u,w)}{n}\log\frac{\uppi_0(u)}{\alpha^{(\ell)}(u)}=n\sum_{\abs{u}=\ell}\pisn^{(\ell)}(u)\log\frac{\uppi_0(u)}{\alpha^{(\ell)}(u)}, \notag\\
 	&=n\sum_{\abs{u}=\ell}\bigg[\pisnl(u)\log\frac{\pisnl(u)}{\alpha^{(\ell)}(u)}-\pisnl(u)\log\frac{\pisnl(u)}{\uppi_0(u)}\bigg]&\notag\\
 	&=n\bigg(D(\pisnl||\al)-D(\pisnl||\uppi_0)\bigg)\label{ineq:diff_of_divg}&
 \end{align}
Then by (\ref{eq:divg_tend_to_zero}) and (\ref{ineq:divg_away_from_zero}), we have 
\begin{align*}
	\log d_{G,\alpha}-\log C_0&\geq	n\bigg(D(\pisnl||\al)-D(\pisnl||\uppi_0)\bigg) \\
	&\geq n\bigg(D(\uppi_0||\al)(1-\epsilon)-D(\pisnl||\uppi_0)\bigg)\geq \frac{\abs{w}}{\ell}D(\uppi_0||\al)(1-2\epsilon).
\end{align*}
Therefore, by (\ref{eq:way_to_pick_pi_zero}) we have
\[
	d_{G,\alpha}(w)\geq C_0 2^{\frac{\abs{w}(1-2\epsilon)}{\ell}D(\uppi_0||\al)}\geq 2^{\abs{w}r(1-2\epsilon)\Divg(S||\alpha)}.
\]
Since $r$ and $1-2\epsilon$ can be picked arbitrary close to $1$, take $r(1-2\epsilon) >\gamma$, then 
\[
	d_{G,\alpha}(w)\geq 2^{\gamma\Divg(S||\alpha)\abs{w}}
\]
for $w=S\restrict n_k$  long enough.

We now prove the second part of the main theorem. 

	Let $S$ be a normal number, $G$ an arbitrary finite-state gambler. By Proposition 2.5 of \cite{SchSti72}, $G=(Q,\delta, s,B)$ will eventually reach to a \emph{bottom strongly connected component} (a component that has no path to leave) when processing $S$. A similar argument can also be found in \cite{shen2017automatic}. Without loss of generality, we will therefore assume that every state in $G$ is recurrent in processing $S$.

	Let $w=a_0\cdots a_{n-1}\prefix S$. Then
\begin{align}
		d_{G,\alpha}(w)&=\frac{B\big(\delta(\lambda)\big)(a_0)\cdots B\big(\delta(w[a_0..a_{n-2}])\big)(a_{n-1})}{\alpha(a_0)\cdots\alpha(a_{n-1})}\notag	\\
		&=\prod_{q\in Q}\prod_{a\in\Sigma}\bigg(\frac{B(q)(a)}{\alpha(a)}\bigg)^{\shpgw(q,a)}\label{eq:capital_after_w},
\end{align}
where the notation $\shpgw(q,a)$ denotes the number of times $G$ lands on state $q$ and the next symbol is $a$ while processing $w$. Similarly, we use the notation $\shpgw(q)$ to denote the number of times $G$ lands on $q$ in the same process. 

Taking the logarithm of both sides of (\ref{eq:capital_after_w}), we have 
\begin{align}
	\log \dga(w)&=\sum_{q\in Q}\sum_{a\in \Sigma}\shpgw(q,a)\log\frac{B(q)(a)}{\alpha(a)}\notag\\
				&=\sum_{q\in Q}\shpgw(q)\sum_{a\in\Sigma}\frac{\shpgw(q,a)}{\shpgw(q)}\log\frac{B(q)(a)}{\alpha(a)}\label{eq:capital_freq_of_state}
\end{align}

By a result of Agafonov \cite{agafonov1968normal}, which extends easily to the arbitrary probability measures considered here, we have that, for every state $q$, the limit of $\frac{\shpgw(q,a)}{\shpgw(s)}$ along $S$ exists and converges to $\alpha(a)$. That is 
\begin{equation}\label{eq:aga_lim}
\lim_{w\to S}\frac{\shpgw(q,a)}{\shpgw(s)}=\alpha(a),
\end{equation}
for every state $q$.

Therefore, by equations (\ref{eq:capital_freq_of_state}) and (\ref{eq:aga_lim}), and the fact that there are finitely many states, we have
\begin{align*}
	\log \dga(w)&\leq\sum_{q\in Q}\shpgw(q)\sum_{a\in\Sigma}(\alpha(a) + o(1))\,\log\frac{B(q)(a)}{\alpha(a)}\\
				&=\sum_{q\in Q}-\risk_G(q)\shpgw(q) + \sum_{q\in Q}\shpgw(q)\sum_{a\in\Sigma}o(1)\,\log\frac{B(q)(a)}{\alpha(a)}\\
				&=-\Risk_G(w)+ \sum_{q\in Q}\shpgw(q)\sum_{a\in\Sigma}o(1)\,\log\frac{B(q)(a)}{\alpha(a)}\label{ineq:risk_error}\\
				&=\Risk_G(w)(-1+o(1)).
\end{align*}

It follows that 
\[
	\dga(w)\leq 2^{- (1+o(1))\Risk_B(w)},
\]
so part 2 of the theorem holds.

\end{proof}

\section{Dimension}
Finite-state dimensions give a particularly sharp formulation of part 1 of the strong dichotomy theorem, along with a dual of this result. 

Finite-state dimensions were introduced for the uniform probability measure on $\Sigma$ in \cite{DLLM04,AHLM07} and extended to arbitrary probability measure on $\Sigma$ in \cite{lutz11}. For each $\alpha\in \tri(\Sigma)$ and each $S\in \sinf$, define the sets
\[
	\mathfrak{G}^{\alpha}(S)=\bigg\{s\in[0,\infty)\bigg\vert\
	(\exists \text{FSG } G) \limsup_{w\to S}\alpha^{\abs{w}}(w)^{1-s}\dga(w)=\infty \bigg\}
\]
and 
\[
	\mathfrak{G}_{\str}^{\alpha}(S)=\bigg\{s\in[0,\infty)\bigg\vert\
	(\exists \text{FSG } G) \liminf_{w\to S}\alpha^{\abs{w}}(w)^{1-s}\dga(w)=\infty \bigg\}
\]
The limits superior and inferior here are taken for successively longer prefixes $w\prefix S$. The ``strong'' subscript of $\mathfrak{G}_{\str}(S)$ refers to the fact that $\alpha^{\abs{w}}(w)^{1-s}\dga(w)$ is required to converge to infinity in a stronger sense than in $\mathfrak{G}^{\alpha}(S)$.
\begin{definition*}[\cite{lutz11}]
Let $\alpha\in \tri(\Sigma)$ and $S\in \sinf$.
	\begin{enumerate}
		\item The \textit{finite-state $\alpha$-dimension} of $S$ is $\dim_{\fs}^{\alpha}(S)=\inf\mathfrak{G}^{\alpha}(S)$.
		\item The \textit{finite-state strong $\alpha$-dimension} of $S$ is $\Dim_{\fs}^{\alpha}(S)=\inf\mathfrak{G}_{\str}^{\alpha}(S)$
	\end{enumerate}
\end{definition*}

It is easy to see that, for all $\alpha\in\tri^+(\Sigma)$ and $S\in \sinf$, $0\leq \dim_{\fs}^{\alpha}(S)\leq \Dim_{\fs}^{\alpha}(S)\leq1.$

\begin{theorem}
	For all $\alpha \in \tri(\Sigma)$ and $S\in \sinf$ let $c= \log(1/ \min_{a\in\Sigma}\alpha(a))$. Then,
	\[
		\dim^{\alpha}_{\fs}(S)\leq 1- \Divg(S||\alpha)/c
	\]
	and 
	\[
		\Dim_{\fs}^{\alpha}(S)\leq 1-\divg(S||\alpha)/c.
	\]
\end{theorem}
\begin{proof}
Let $t<\Divg(S||\alpha)/c$, and let $s=1-t$. Fix $l$ such that $\Divsal/l>tc$, then for almost every $n$,
$D(\pisnl||\alpha^{(\ell)})> l t c$. Note that,  $\alpha^{|w|}(w)\ge
(1/2^{c})^{|w|}$ for every $w\in\sstar$.

Define the gambler $G$ be $G=(Q,\delta, s_0, B_n)$, where
$Q=\Sigma^{\leq\ell-1}$,
\begin{equation*}\begin{aligned}[c]
\delta(w,a)=\begin{cases}
    wa &\quad \text{if $\abs{wa}<\ell$} \\
    \lambda&\quad \text{if $\abs{wa}=\ell$}
\end{cases}
\end{aligned}
\end{equation*}
$ s_0=\lambda$,
 and
$ B_n(w)(a)= \pisnl(a|w)$, where $\pisnl(a | w) $ describes the conditional probability (induced by $\pisnl$) of occurrence of an
$a$ after $w\in Q$.

Let $u=a_0\cdots a_{\ell -1}$ be in $\Sigma^{\ell}$.
\[
    \frac{B_n(\lambda)(a_0)\cdots B_n(u[0..\ell-2])(a_{\ell-1})}{\alpha(a_0)\cdots\alpha(a_{\ell-1})}=\frac{\pisnl(u)}{\alpha^{(\ell)}(u)}.
\]

Then for $z\in\sstar$ with $z\sqsubseteq S$ and $|z|=ln$, we have

\begin{align*}
\alpha^{\abs{z}}(z)^{1-s}\dga(z) &= \alpha^{\abs{z}}(z)^{t}\dga(z) \\
&=\alpha^{\abs{z}}(z)^{t}\prod_{u\in\Sigma^{\ell}}\left(\frac{\pisnl(u)}{\alpha^{(\ell)}(u)}\right)^{n\pisnl(u)}
\end{align*}

Therefore,
\begin{align*}
\alpha^{\abs{z}}(z)^{t} \dga(z) &\ge \frac{1}{2^c}^{|z|t}2^{n D(\pisnl||\alpha^{(\ell)})}\\
&\ge 2^{-c |z|t+c|z|t}
\end{align*}

Since the number of states is fixed,  this implies $\dim^{\alpha}_{\fs}(S)\leq 1- \Divg(S||\alpha)/c$.

The proof of the other case is similar, where we use the fact that, for infinitely many $n$,
$D(\pisnl||\alpha^{(\ell)})> l t c$.
\end{proof}

    \section*{Acknowledgments}
    We thank students in Iowa State University's fall, 2017, advanced topics in computational randomness course for listening to a preliminary version of this research. We thank three anonymous reviewers of an earlier draft of this paper for useful comments and corrections.

    \bibliographystyle{plain}
    \bibliography{master}



\end{document}